\documentclass[reqno,12pt,letterpaper]{amsart}

\usepackage{amsmath,amssymb,amsthm,mathtools,microtype}
\usepackage[hidelinks]{hyperref}

\allowdisplaybreaks

\def\arXiv#1{\href{http://arxiv.org/abs/#1}{arXiv:#1}}
\usepackage{array}
\newcolumntype{P}[1]{>{\centering\arraybackslash}m{#1}}

\DeclareMathOperator{\tr}{tr}
\DeclareMathOperator{\Spec}{Spec}
\newcommand{\RR}{\mathbb R}
\newcommand{\cL}{\mathcal L}
\newcommand{\Lone}{\mathcal L_1}
\newcommand{\Ltwo}{\mathcal L_2}
\newcommand{\la}{\langle}
\newcommand{\ra}{\rangle}

\newtheorem*{theo*}{Theorem}

\newtheorem{theo}{Theorem}
\newtheorem{prop}{Proposition}
\newtheorem{lemm}[prop]{Lemma}

\numberwithin{equation}{section}

\title[Optimal relaxation for Witten Lindbladians]
{OPTIMAL RELAXATION FOR WITTEN LINDBLADIANS IN 1D}
\author{Simon Becker and Maciej Zworski}
\RequirePackage{xcolor} 
\providecommand{\DIFdeltex}[1]{} 
\providecommand{\DIFdel}[1]{\texorpdfstring{\DIFdeltex{#1}}{}} 

\renewcommand{\DIFdel}[1]{}
\RequirePackage{listings} 
\lstdefinelanguage{DIFcode}{ 
  moredelim=[il][\color{white}\tiny]{\%DIF\ <\ }, 
  moredelim=[il]{\%DIF\ >\ } 
} 
\lstdefinestyle{DIFverbatimstyle}{ 
	language=DIFcode, 
	basicstyle=\ttfamily, 
	columns=fullflexible, 
	keepspaces=true 
} 
\lstnewenvironment{DIFverbatim}{\lstset{style=DIFverbatimstyle}}{} 
\lstnewenvironment{DIFverbatim*}{\lstset{style=DIFverbatimstyle,showspaces=true}}{} 

\begin{document}

\begin{abstract}
We prove optimal trace-norm relaxation for the one-dimensional
Lindbladian associated with the Witten differential
$a=h\partial_x+V'$  which annihilates the classical Gibbs  density: if  $\lambda_1(h)$ is the first positive eigenvalue of
{$H=a^*a$} , then the relaxation rate is
$\gamma_h=\lambda_1(h)/(2h)$.  It applies to initial operators 
whose Schwartz kernels, after a 
Gibbs conjugation, satisfy $ L^2 $ estimates for the restriction to the diagonal and for the normal derivative. An appendix by Chat GPT 6 
presents a stronger result specialised to positive initial data. 
\end{abstract}

\maketitle

\section{Statement of the result and motivation}
\label{s:int}

With motivation and references provided below, we study the evolution governed by
the following superoperator:
 \begin{equation}
\cL T=\frac{1}h\big(aTa^*-\tfrac12 ( a^*a T + T a^*a ) \big), \ \ \ 
 a := h \partial_x +V' ( x ) . 
 \label{eq:L}
\end{equation}
Since the jump operator, $ a $, is given by a Witten differential
\cite{wit} (and $ H:=a^* a $ is the scalar Witten Laplacian), we refer to $ \cL $ as the {\em Witten Lindbladian}. The assumptions on $ V $ are given in \eqref{eq:assV}. By adding an irrelevant $ h $-dependent constant we assume that 
\begin{equation}
\label{eq:assV0}
 \begin{gathered}
 \int_{\mathbb R}e^{-2V(x)/h}\,dx=1.
 \end{gathered}
\end{equation} 
For initial conditions for the Lindblad evolution, we consider $T\in\mathcal L_1(L^2(\mathbb R))$, the space of trace-class operators on $L^2(\mathbb R)$, 
\[  T u ( x ) = \int_{\mathbb R }  e^{-(V(x)+V(y))/h}r(x,y) u ( y ) dy ,  \ \ \ r \in C^1 , \]
where, with $ d\mu(y):=e^{-2V(y)/h}\,dy$,  $ r $ satisfies
\begin{equation}
\label{eq:assr}
 \begin{gathered}
 r|_{\{(x,x)\}}\in L^2(d\mu),\ \ \ \ \ 
 |\partial_y r(x,y)|\leq G_0(y) \in L^2(d\mu). 
 \end{gathered}
\end{equation}
\begin{theo}
\label{t:main}
There is a constant $C_{h,V}>0$ such that for $ T $ above with $\tr T = 1 $, we have for 
 $ t \geq 0 $ and $\Pi:=|e^{-V/h}\rangle\langle e^{-V/h}|$, 
\begin{equation}
 \|e^{ t \mathcal L } T - \Pi \|_{\Lone}
 \leq C_{h,V}e^{-\gamma_h t}
 \left(
   \|T\|_{\Lone}
   +\|r|_{ \{ (x,x) \} } - 1   \|_{L^2(\mu)}
   +\|G_0\|_{L^2(\mu)}
 \right), 
 \label{eq:main}
\end{equation}
where
 \begin{equation} 
\label{eq:defgam} \gamma_h:=\lambda_1(h)/(2h)>0, \ \ \ \
\lambda_1(h):=\inf(\Spec(H)\setminus\{0\}) , \end{equation} and this exponent is {\em optimal}.
\end{theo}

\noindent
{\bf Remarks.} 1.
The regularity assumptions on $r$ are chosen to keep the proof
elementary.  The pointwise differentiability in $y$ is likely to be
replaceable by an appropriate absolute-continuity condition.  We do not
pursue this endpoint question here, but some restrictions on $T$ are always needed, as $\mathcal L$ does not have a spectral gap on 
$ \mathcal L_1 $ (easily seen in the case of $ V ( x ) =x^2$.  
The more substantial problem is
the extension to higher dimensions: our methods are very much one
dimensional.

\noindent
2. Tracing the constants in the proof, one may take
\[
 C_{h,V}\leq
 \max\!\left(
 2e^{\gamma_h},1,
 2A_{h,V}(1+\gamma_h)^{1/2}\mathcal V(V)
 \right),
\]
where (with $ S $ defined in \eqref{eq:QS})
\[
 A_{h,V}
 =\|\langle x\rangle(1-S)^{-1/2}\|_{L^2(\mu)\to L^2(\mu)},
\]
and $ \mathcal V ( V )$ was defined in \eqref{eq:Vconstant}. 
For $0<h\leq h_0$, the assumptions on $V$ give
$A_{h,V}=\mathcal O_V(1)$ and $\gamma_h= \mathcal O_V(1)$.  Thus the $h$-dependence
provided by the proof is contained in $\mathcal V(V)$.  It is
bounded as $h\to0$ in the one-well case, while for a nondegenerate
double well of barrier height $S_0$ the proof gives
\[
 C_{h,V}=\mathcal O_V(\sqrt h\,e^{S_0/h})
       = \mathcal O_V\!\left(\sqrt{h/\gamma_h}\right).
\]
We do not claim that this is the optimal common prefactor in the
estimate.

\noindent
3. When, in the theorem, $T\geq0$, we can replace the hypotheses by
$r|_{\{(x,x)\}}\in L^\infty$ and the assumption that $ \tr HT = \| a T^{\frac12} \|_2^2 < \infty $. The result then follows from
the following inequality proved at the of this introduction:
\begin{equation}
 |\partial_y r(x,y)|\leq G_0(y),\ \ \ \  \|G_0\|_{L^2(\mu)}
 \leq h^{-1}\|r|_{\{ (x,x)\} } \|_{L^\infty}^{\frac12} \| a T^{\frac{1}{2}} \|_{\Ltwo} .
 \label{eq:positive-majorant}
\end{equation}
At our request, Chat GPT 6 upgraded this by weakening the hypothesis to
$r|_{\{(x,x)\}}\in L^q(d\mu)$, $q>1$.  That result, Theorem \ref{t:positive-states},  is presented in the appendix.

\medskip

When $V(x)=x^2/2$, \eqref{eq:L} is, after the standard normalization of
the annihilation operator, the one-mode pure-loss channel, with
transmissivity $e^{-2t}$; see
\cite{HoWe01,Weedbrook12,DPTG17,DPTG18}.  The general potential replaces
the harmonic-oscillator annihilation operator by the Witten differential
$h\partial_x+V'$.  Related first-order factorizations for anharmonic
oscillators appear in \cite{Gao99,NestSaalfrank00,Nest02}.

Lindblad generators \cite{GKSL,Lindblad} provide the standard Markovian
framework for dissipative quantum dynamics.  Recently, they have also
appeared in quantum Gibbs sampling; see \cite{chen,DLL} for finite-dimensional
constructions and \cite{gaz7,beck} for related infinite-dimensional work.
A model based on the simpler equation considered here appeared in
Leng--Ding--Chen--Lin \cite{LDCL}, where Witten differentials were used as
Lindblad jump operators.  The stationary state $\Pi$ is obtained from the
square root $e^{-V/h}$ of the Gibbs density.  The trace norm is natural here since
\[
 |\tr((T-\Pi)O)|\leq \|T-\Pi\|_{\Lone}\|O\|
\]
for every bounded observable $O$.

The diagonal evolution is governed by the Witten Laplacian $H$.  The
theorem says that the full Lindblad evolution, including the off-diagonal
part of the kernel, has the same optimal relaxation rate $\gamma_h$.  For a
self-contained one-dimensional account of the semiclassical asymptotics of
$\lambda_1(h)$, and references, see \cite{MZ}.  If $V$ has exactly one critical point,  a unique nondegenerate minimum
$x_0$, harmonic approximation gives
\begin{equation}
 \gamma_h=V''(x_0)+\mathcal O(h).
 \label{eq:harmonic-gamma}
\end{equation}
For a  multiwell potential, $\gamma_h$ is exponentially small and is governed
by the Eyring--Kramers law.  Thus the result includes the metastable regime.
Related semiclassical Lindblad--Fokker--Planck constructions are discussed in
\cite{hrr,gaz5,kevin,smith}.

We now impose, in addition to the normalization \eqref{eq:assV0},
\begin{equation}
\label{eq:assV}
 \begin{gathered}
 V^{(k)} ( x) = \mathcal O ( \langle x \rangle ), \ k \geq 1 , \ \ \  \ V''(x)\geq-C_0 ,\ \ \ \ xV'(x)\geq x^2/C_1-C_1.
 \end{gathered}
\end{equation} 
These assumptions are chosen for a short proof rather than for maximal
generality.  Since $-H/(2h)$ generates a contraction semigroup on
$L^2(\mathbb R)$, the Davies construction gives the trace-class contraction
semigroup used below; see \cite[\S 1, Remark]{gaz5}.  The remaining bounds are
used in the maximum-principle, weighted, and tail estimates.  We write
\begin{equation}
 H:=a^*a=-h^2\partial_x^2+|V'|^2-hV'',\qquad
 \nu:=e^{-V/h},\qquad \Pi:=|\nu\ra\la\nu|.
 \label{eq:Hnu}
\end{equation}
Then, $a\nu=0$,  $\ker H=\mathbb C\nu$, and,  from \eqref{eq:assV0},  $\|\nu\|_{L^2}=1$.  The assumptions
\eqref{eq:assV} imply that $H$ has compact resolvent and hence that
$\lambda_1(h)>0$.

\begin{proof}[Proof of \eqref{eq:positive-majorant}]
By the spectral theorem, 
$ T=\sum_{j\geq0 }\lambda_j|\varphi_j\ra\la\varphi_j|,$
$ \lambda_j\geq0 $, 
$\sum_{j\geq 0 }\lambda_j=1 $, 
where the $\varphi_j$ are orthonormal in $L^2(\mathbb R)$.  Writing
$\varphi_j=\nu f_j$ we introduce the $\ell^2$-valued function
$ 
 F(x):= \{ \sqrt{\lambda_j} f_j(x) \} _{j\in \mathbb  N } \in   L^2 ( \mathbb R ; \ell^2 ( \mathbb N ), d \mu  ) $. 
It follows that \begin{equation}
 r(x,y)=\sum_{j\geq0} \lambda_jf_j(x)\overline{f_j(y)}
       =\la F(x),F(y)\ra_{\ell^2},
 \qquad r ( x, x)=\|F(x)\|_{\ell^2}^2 .  
 \label{eq:Gfact}
\end{equation}

The definition \eqref{eq:Hnu}, the spectral decomposition above, and the fact that
$a(\nu f)=h\nu f'$ give 
\begin{equation}
\label{eq:E2F}  \begin{split}   \tr HT & = \|aT^{\frac12}\|_{\Ltwo}^2 =
\sum_{j = 0 }^\infty \lambda_j\|a\varphi_j\|_{L^2 ( \mathbb R ) }^2 \\
& = \sum_{j=0}^\infty \lambda_j \| h \nu f'_j \|_{L^2 ( \mathbb R ) }^2 = 
h^2 \int_{\mathbb R } \| F' ( x ) \|_{\ell^2 ( \mathbb N ) }^2 d \mu ( x ) . 
\end{split} \end{equation}
Differentiating \eqref{eq:Gfact} in the weak sense gives
$ 
 \partial_y r(x,y)=\la F(x),F'(y)\ra_{\ell^2}$. 
Hence, by the Cauchy--Schwarz inequality and \eqref{eq:Gfact}, 
\begin{equation}
 |\partial_y r(x,y)|^2
 \leq r(x,x)\|F'(y)\|_{\ell^2}^2
 \leq\|r|_{ \{ (x,x) \} } \|_{L^\infty}\|F'(y)\|_{\ell^2}^2.
 \label{eq:derivative-CS}.
\end{equation}
This and \eqref{eq:E2F} proves \eqref{eq:positive-majorant}. 
\end{proof}

\medskip
\noindent
{\sc Acknowledgements.} This note grew out of discussions between the authors and ChatGPT 5.6 about improving 
an earlier probabilistic argument with a non-optimal decay rate  $ \gamma_h/2 $ \cite{Be26}. In particular, the key Volterra factorisation argument was suggestedby ChatGPT and developed to a substantial extent through these discussions. MZ gratefully acknowledges partial support from the Simons
Foundation through Targeted Grant Award No.~896630, ``Moir\'e Materials
Magic''.

\section{The conjugated kernel and the maximum principle}

In the language of the theorem in \S \ref{s:int}, let $T_0=T$ and $T_t=e^{t\cL}T_0$.  
Using \eqref{eq:Hnu}, write
\begin{equation}
 K_t(x,y)=\nu(x)\nu(y)r_t(x,y).
 \label{eq:factor}
\end{equation}
The identities
 \[
a(\nu f)=h\nu f',\quad
 a ^* a (\nu f)=h\nu(2V'f'-hf'') \]
give the following simple
\begin{lemm}
\label{l:kernel}
The conjugated kernel satisfies
\begin{equation}
 \partial_t r_t=Mr_t,\qquad
 M=\tfrac 12h(\partial_x+\partial_y)^2
   -V'(x)\partial_x-V'(y)\partial_y.
 \label{eq:M}
\end{equation}
If $Rr(x):=r(x,x)$ and
\begin{equation}
 Q=\tfrac 12 h\partial_x^2-V'\partial_x,\qquad S=Q-V'',
 \label{eq:QS}
\end{equation}
then
\begin{equation}
 RM=QR,\qquad
 \partial_yM=(M-V''(y))\partial_y.
 \label{eq:intertwining}
\end{equation}
\end{lemm}

One striking consequence of \eqref{eq:intertwining} is 
\begin{equation}
 r_t(x,x)=e^{tQ}r_0(x,x).
 \label{eq:diagflow}
\end{equation}
Multiplication by $\nu$ gives the unitary equivalences of (unbounded) operators
on $ L^2 ( d\mu ) $ (on the left) with operators on $ L^2 ( dx ) $ (on the right):
\begin{equation}
Q\simeq- \frac{a^*a }{2h},\qquad
 S\simeq-\frac{aa^*}{2h}.
 \label{eq:unitary}
\end{equation}
The operator $aa^*$ has no zero mode, and its spectrum agrees with the
nonzero spectrum of $H = a^* a $.  Hence
\begin{equation}
 \|e^{tQ}f-\mu(f)\|_{L^2(\mu)}
 \leq e^{-\gamma_ht}\|f-\mu(f)\|_{L^2(\mu)},
 \label{eq:Qgap}
\end{equation}
\begin{equation}
 \|e^{tS}g\|_{L^2(\mu)}
 \leq e^{-\gamma_ht}\|g\|_{L^2(\mu)}.
 \label{eq:Sgap}
\end{equation}
Since $\int r_0(x,x) d\mu ( x )  =\tr T_0=1$, this gives
\begin{equation}
 \|r_t(x,x)-1\|_{L^2(\mu)}
 \leq e^{-\gamma_ht}\|r_0(x,x)-1\|_{L^2(\mu)}.
 \label{eq:diagdecay}
\end{equation}

For the whole-space comparison we will use the second condition in \eqref{eq:assr}
and the maximum principle for degenerate heat equations. Thus we start with
\begin{lemm}
\label{l:comparison}
Let $w \in C^2 ([0,T]\times\mathbb R^2; \mathbb R ) $ have at 
most polynomial growth, and satisfy
\begin{equation}
 (\partial_t-M+V''(y))w\geq0,\qquad w(0,x,y)\geq0.
 \label{eq:comparison-eq}
\end{equation}
Then $w ( t, x, y ) \geq0$ for $ (t,x,y) \in [0,T ] \times \mathbb R^2 $.
\end{lemm}

\begin{proof}
For $C_0$ in \eqref{eq:assV} put 
\[ z=e^{-C_0t}w,\qquad q(y)=-V''(y)-C_0\leq0,
 \qquad P=\partial_t-M-q(y).
\]
Then \eqref{eq:comparison-eq} gives $Pz\geq0$.  Let $\rho(x,y)=1+x^2+y^2$.  For every integer $k\geq1$,
\eqref{eq:assV} and direct differentiation give $
 M\rho^k\leq C_k\rho^k$ (see \eqref{eq:M} for the definition of $M $).

Choose $m$ so that $|z(t,x,y)|\leq C_T\rho(x,y)^m$ on $[0,T]\times\mathbb R^2$
and put
\[
 \varphi(t,x,y)=e^{At}\rho(x,y)^{m+1}.
\]
Since $q\leq0$, choosing $A>C_{m+1}$ gives $P\varphi>0$.

For $\varepsilon>0$, set $v=z+\varepsilon\varphi$.  It is positive at $t=0$, and it tends to $+\infty$
as $x^2+y^2\to\infty$, uniformly for $0\leq t\leq T$.
If it first vanished at $(t_0,x_0,y_0)$, then
\[
 \partial_tv\leq0,\qquad \partial_xv=\partial_yv=0,
 \qquad \nabla^2_{x,y}v\geq0.
\]
In particular,
$(\partial_x+\partial_y)^2v\geq0$.  At the contact point the drift and
zeroth-order terms vanish, and therefore
\[
 Pv=\partial_tv-\tfrac12 h (\partial_x+\partial_y)^2v\leq0.
\]
But $Pv=Pz+\varepsilon P\varphi>0$, a contradiction.  Thus
$z+\varepsilon\varphi>0$ and letting $\varepsilon\to0+$ proves the claim.
\end{proof}

The next proposition propagates the estimate \eqref{eq:assr} on the initial
condition:
\begin{prop}
\label{p:normal}
For all $t\geq0$, and for $ S $ defined in \eqref{eq:QS},
\begin{equation}
 |\partial_y r_t(x,y)|\leq(e^{tS}G_0)(y),\qquad x,y\in\mathbb R.
 \label{eq:maximum-principle}
\end{equation}
\end{prop}

\begin{proof}
Put $u_t=\partial_y r_t$.  By \eqref{eq:intertwining},
\[
 \partial_tu_t=(M-V''(y))u_t.
\]
Let $g_t=e^{tS}G_0$ and regard $g_t(y)$ as a function of $(x,y)$.  Since $M$
reduces to $Q$ on functions of $y$ alone, it satisfies the same equation.
For bounded smooth data, Lemma~\ref{l:comparison}, applied for every
$\theta\in\mathbb R$ to
\[
 g_t(y)-\operatorname{Re}(e^{-i\theta}u_t(x,y)),
\]
gives \eqref{eq:maximum-principle}.  The same lemma shows that the closed
semigroup generated by $M-V''(y)$ is positivity preserving.  Since
$u_0$ and $G_0(y)$ belong to $L^2(\mathbb R^2 ; d\mu(x)d\mu(y))$, the order inequalities
\[
 -G_0(y)\leq\operatorname{Re}(e^{-i\theta}u_0(x,y))\leq G_0(y)
\]
pass to the $L^2$ closure.  This proves the result for the stated
$G_0\in L^2(\mu)$.
\end{proof}

For quadratic confinement the Volterra estimate below requires one spatial
weight.  It is obtained using smoothing for positive times:
\begin{lemm}
\label{l:smoothing}
There is $C_{h,V}>0$ such that, for $g\in L^2(\mu)$ and $t>0$,
\begin{equation}
 \|\la x\ra e^{tS}g\|_{L^2(\mu)}
 \leq C_{h,V}(1+t^{-\frac{1}{2}})e^{-\gamma_ht}\|g\|_{L^2(\mu)}.
 \label{eq:weighted-smoothing}
\end{equation}
\end{lemm}

\begin{proof}
As stated in \eqref{eq:unitary}, $1-S$ is unitarily equivalent to
$1+aa^*/(2h)$.  The assumptions \eqref{eq:assV} imply
$|V'(x)|^2\geq cx^2-C$, while
\[
 aa^*=-h^2\partial_x^2+|V'|^2+hV''.
\]
Hence its quadratic form gives
\begin{equation}
 \|\la x\ra f\|_{L^2(\mu)}
 \leq C_{h,V}\|(1-S)^{1/2}f\|_{L^2(\mu)}.
 \label{eq:graph-weight}
\end{equation}
The spectral theorem and the fact that $\Spec(S)\subset(-\infty,-\gamma_h]$ give
\[
 \|(1-S)^{1/2}e^{tS}\|_{L^2(\mu)\to L^2(\mu)}
 \leq C_{h,V}(1+t^{-1/2})e^{-\gamma_ht}.
\]
Combining the two estimates proves the lemma.
\end{proof}

\section{The Volterra factorisation and proof of Theorem \ref{t:main}}

Define
\begin{equation}
 F_-(s)=\int_{-\infty}^s e^{-2V(x)/h}\,dx,
 \qquad F_+(s)=\int_s^\infty e^{-2V(x)/h}\,dx,
 \label{eq:tails}
\end{equation}
and
\begin{equation}
 \mathcal V(V)^2
 =\int_{\mathbb R}F_-(s)F_+(s)e^{2V(s)/h}\la s\ra^{-2}\,ds.
 \label{eq:Vconstant}
\end{equation}

\begin{lemm}
\label{l:tails}
For $ V $ satisfying \eqref{eq:assV}, $\mathcal V(V)<\infty$.
\end{lemm}

\begin{proof}
For $s$ sufficiently large, \eqref{eq:assV} gives $V'(x)\geq cx$ for
$x\geq s$.  Therefore
\[
 F_+(s)\leq C_{h,V}s^{-1}e^{-2V(s)/h}.
\]
Since $F_-(s)\leq1$, the integrand in \eqref{eq:Vconstant} is
$\mathcal O(s^{-3})$ as $s\to+\infty$.  The argument for 
$ s \to - \infty $  is identical.
\end{proof}

The next lemma provides the crucial Volterra trace factorisation inspired by 
\cite{GK}:

\begin{lemm}
\label{l:volterra}
Suppose that $|u(x,s)|\leq v(s)$ and $\la s\ra v\in L^2(\mu)$.  Then the
operator with kernel
\begin{equation}
 R(x,y):=\nu(x)\nu(y)\int_x^y u(x,s)\,ds
 \label{eq:O}
\end{equation}
is of trace class and
\begin{equation}
 \|R\|_{\Lone}\leq2\mathcal V(V)\|\la s\ra v\|_{L^2(\mu)}.
 \label{eq:Volterra}
\end{equation}
\end{lemm}

\begin{proof}
We write 
\[ R = R_+ + R_-, \ \ \ \ \  R_\pm ( x, y ) := R ( x, y) \mathbf 1_{ \pm x < \pm y }, \]
and for the factorisation of $ R_+$, we define
\[
 w_+(s)=F_-(s)e^{2V(s)/h}\la s\ra^{-2} > 0 . 
\]
We then put 
\[
 A_+(x,s) :=\nu(x)u(x,s){\bf1}_{\{x<s\}}w_+(s)^{-1/2}, \ \ \ 
 B_+(s,y) : =w_+(s)^{1/2}\nu(y){\bf1}_{\{s<y\}}, 
\]
so that (identifying integral kernels with operators),  $ R_+ = A_+ B_+ $. We now have 
\begin{equation}
 \|A_+\|_{\Ltwo}^2
 \leq\int w_+(s)^{-1}v(s)^2F_-(s)\,ds
 =\|\la s\ra v\|_{L^2(\mu)}^2,
 \label{eq:Aplus}
\end{equation}
\begin{equation}
 \|B_+\|_{\Ltwo}^2
 =\int w_+(s)F_+(s)\,ds=\mathcal V(V)^2.
 \label{eq:Bplus}
\end{equation}
This gives $ \| R_+ \|_{\Lone} \leq \mathcal V(V)\|\la s\ra v\|_{L^2(\mu) } $.
For $R_-$, put
$w_-(s)=F_+(s)e^{2V(s)/h}\la s\ra^{-2}$ and
\[
 A_-(x,s):=\nu(x)u(x,s){\bf1}_{\{s<x\}}w_-(s)^{-1/2},\qquad
 B_-(s,y):=w_-(s)^{1/2}\nu(y){\bf1}_{\{y<s\}}.
\]
Then $R_-=-A_-B_-$, and the same estimates give
$\|R_-\|_{\Lone}\leq \mathcal V(V)\|\la s\ra v\|_{L^2(\mu)}$.
Together with the estimate for $R_+$ this gives \eqref{eq:Volterra}.
\end{proof}

\begin{proof}[Proof of the Theorem \ref{t:main}]
Put $B_t=T_t-\Pi$.  For $t>0$, split its kernel as
\begin{equation}
\begin{split}
 B_t(x,y) & =\nu(x)\nu(y)(r_t(x,x)-1)
+\nu(x)\nu(y)(r_t(x,y)-r_t(x,x))\\
 & =:D_t(x,y)+R_t(x,y).
\end{split}
\label{eq:split}
\end{equation}
The first term has rank one, and \eqref{eq:diagdecay} gives
\begin{equation}
 \|D_t\|_{\Lone}
 \leq e^{-\gamma_ht}\|r_0(x,x)-1\|_{L^2(\mu)}.
 \label{eq:Dtrace}
\end{equation}
For the second term,
\[
 r_t(x,y)-r_t(x,x)=\int_x^y\partial_y r_t(x,s)\,ds.
\]
By Proposition~\ref{p:normal}, the integrand is dominated by $v_t=e^{tS}G_0$.  For
$t\geq1$, Lemmas~\ref{l:smoothing} and \ref{l:volterra} give
\begin{equation}
 \|R_t\|_{\Lone}
 \leq C_{h,V}e^{-\gamma_ht}\|G_0\|_{L^2(\mu)}.
 \label{eq:Otrace}
\end{equation}
This proves \eqref{eq:main} for $t\geq1$.  For $0\leq t\leq1$, stationarity
of $\Pi$ and trace-norm contractivity give
\[
 \|B_t\|_{\Lone}\leq\|T_0-\Pi\|_{\Lone}
 \leq\|T_0\|_{\Lone}+1,
\]
and the short interval is absorbed by enlarging $C_{h,V}$.

For optimality, let $H\varphi_1=\lambda_1(h)\varphi_1$, with
$\varphi_1\perp\nu$, and write $\varphi_1=\nu f_1$.  Then
$f_1\in L^2(\mu)$ and $f_1'\in L^2(\mu)$. To see the second claim we note that
\begin{equation*}
  h^2\|f_1'\|_{L^2(d\mu)}^2
 =\|a(\nu f_1)\|_{L^2(dx)}^2
 =\langle H\varphi_1,\varphi_1\rangle_{L^2(dx)}
 =\lambda_1(h)\|\varphi_1\|_{L^2(dx)}^2<\infty .
\end{equation*}
This shows that operator
\[
 T_0=\Pi+|\nu\ra\la\varphi_1|, \ \ \  \ \tr T_0 = 1 +\langle \nu , \varphi_1 \rangle_{L^2( \mathbb R )} = 1, 
\]
has conjugated kernel $r_0(x,y)=1+\overline{f_1(y)}$, so that
$r_0(x,x)-1=\overline{f_1(x)}\in L^2(\mu)$ and
$|\partial_y r_0(x,y)|=|f_1'(y)|$.  Hence it
satisfies the hypotheses of the theorem, and
\[
 \cL(|\nu\ra\la\varphi_1|)=-\gamma_h|\nu\ra\la\varphi_1|.
\]
Thus no larger uniform exponent is possible.
\end{proof}

\renewcommand{\theequation}{A.\arabic{equation}}
\refstepcounter{section}
\renewcommand{\thesection}{A}
\setcounter{equation}{0}

\section*{Appendix by ChatGPT 6}

The following strengthens Remark~3 in \S\ref{s:int}, under the same
assumptions \eqref{eq:assV0} and \eqref{eq:assV}.
Function norms $\|\cdot\|_p$ are in $L^p(\mu)$, and inner products
are linear in the first factor.

\begin{theo}\label{t:positive-states}
 Let $T\geq0$, $\tr T=1$, and let $d(x)=r(x,x)$ be its conjugated
diagonal, as in \eqref{eq:Gfact}. If 
\[ E(T): = \tr HT = \| a T^{\frac12} \|_{\mathcal L_2 }  < \infty  \ \text{ and } \  d\in L^q(d\mu), \ \ 1<q<\infty,  \]
 then
\begin{equation}
 \|e^{t\cL}T-\Pi\|_{\Lone}
 \leq C_{h,V,q}e^{-\gamma_ht}
 \left(\frac{\|d\|_q E(T)}{\lambda_1(h)}\right)^{1/2},
 \qquad t\geq0.
 \label{eq:app-result}
\end{equation}
For $d\in L^\infty(\mu)$, replace $\|d\|_q$ by $\|d\|_\infty$;
the estimate then holds with an absolute constant.
 \end{theo}

\begin{proof}
Let $F(x)=\{\sqrt{\lambda_j}f_j(x)\}_j$ be the $\ell^2$-valued
function introduced in the proof of \eqref{eq:positive-majorant}.
Equations \eqref{eq:Gfact}--\eqref{eq:E2F} give
\[
 r(x,y)=\la F(x),F(y)\ra,\qquad d(x)=\|F(x)\|^2,
 \qquad h^2\|F'\|_2^2=E(T).
\]
We first suppose that $F\in C_c^\infty(\RR;\mathbb C^N)$ and
$\|F\|_2=1$; the approximation is given below.

For any positive trace-one operator $A$, the operator $A-\Pi$
has at most one negative eigenvalue, since its quadratic form is
nonnegative on $\nu^\perp$. Its trace is zero, so
\begin{equation}
 \|A-\Pi\|_{\Lone}=2\|A-\Pi\|\leq2\|A-\Pi\|_{\Ltwo};
 \label{eq:app-traceHS}
\end{equation}
see also \cite{CCC}. It therefore suffices to estimate the
Hilbert--Schmidt norm. To this end, put
\begin{equation}
 b(x)=\int_0^x2\|F(s)\|\,\|F'(s)\|\,ds,\qquad
 k(x,y)=r(x,y)-\tfrac12(d(x)+d(y)).
 \label{eq:app-b}
\end{equation}
Then $b$ is bounded and nondecreasing, and
\begin{equation}
 |k(x,y)|\leq |b(x)-b(y)|,\qquad
 |d(x)-d(y)|\leq |b(x)-b(y)|.
 \label{eq:app-initial}
\end{equation}
Indeed, for complex Hilbert-space vectors $u,v$, with
$\alpha=\|u\|$, $\beta=\|v\|$,
\[
 \left|\la u,v\ra-\tfrac12(\alpha^2+\beta^2)\right|
 \leq\tfrac12(\alpha+\beta)\|u-v\|
 \leq\|\alpha u-\beta v\|.
\]
The first inequality follows by Cauchy--Schwarz, the second by expanding
squared norms. Apply them to $u=F(x)$, $v=F(y)$ and use
$\|(\|F\|F)'\|\leq2\|F\|\|F'\|=b'$; also $|d'|\leq b'$.

For $p=2q/(q+1)\in(1,2)$, H\"older also gives
\begin{equation}
 \|b'\|_p\leq2\|F\|_{2q}\|F'\|_2
 =\frac{2}h\sqrt{\|d\|_q E(T)}.
 \label{eq:app-holder}
\end{equation}

To propagate \eqref{eq:app-initial}, use $M,Q,S$ from
\eqref{eq:M}--\eqref{eq:QS} and put $P_t=e^{tQ}$, $d_t=P_td$, $b_t=P_tb$,
$k_t=r_t-\tfrac12(d_t(x)+d_t(y))$; $d_t(x)=r_t(x,x)$ by \eqref{eq:diagflow}.
On $x<y$, both $k_t$ and $b_t(y)-b_t(x)$ solve $\partial_tu=Mu$
and vanish on $x=y$. The barrier proof of Lemma~\ref{l:comparison},
with the zeroth-order term omitted, applied to
$b_t(y)-b_t(x)-\operatorname{Re}(e^{i\theta}k_t)$ for every $\theta$,
propagates the first bound. Apply the same argument to
$d_t(y)-d_t(x)$ for the second bound, and reflect for $x>y$.
Smooth majorants converging uniformly to $b'$ justify comparison
when $b'$ is not smooth.

Set $\sigma_t=\|b_t-\mu(b_t)\|_2$. Using
$2\|f-\mu(f)\|_2^2=\iint|f(x)-f(y)|^2\,d\mu(x)d\mu(y)$ gives
\[
 \|d_t-1\|_2\leq\sigma_t,\qquad
 \|k_t\|_{L^2(\mu\otimes\mu)}\leq\sqrt2\,\sigma_t.
\]
The $L^2(\mu\otimes\mu)$ norm of
$\tfrac12(d_t(x)+d_t(y))-1$ is $\|d_t-1\|_2/\sqrt2$.
Since $\partial_xQ=S\partial_x$, we have $b_t'=e^{tS}b'$.
Thus \eqref{eq:app-traceHS} and the Poincar\'e inequality associated with
\eqref{eq:unitary}--\eqref{eq:Qgap} give
\begin{equation}
 \|e^{t\cL}T-\Pi\|_{\Lone}
 \leq3\sqrt2\,\sigma_t
 \leq\frac{3\sqrt2\,h}{\sqrt{\lambda_1(h)}}
       \|e^{tS}b'\|_2.
 \label{eq:app-key}
\end{equation}

The confinement assumptions \eqref{eq:assV} give $c|x|\leq |V'(x)|\leq C|x|$,
$xV'(x)>0$, near infinity. Writing $\rho=e^{-2V/h}$ and defining $ m $ to the median of $ \mu $ ($\int_\mu^\infty 
\rho ( s ) ds = \frac12 $), these bounds give,
for large $x>0$,
\[
 F_+(x)\asymp_{h,V}\frac{\rho(x)}x,\qquad
 \int_m^x\frac{ds}{\rho(s)}\lesssim_{h,V}\frac{1}{x\rho(x)},\qquad
 \log\frac{1}{F_+(x)}\lesssim_{h,V}x^2.
\]
Hence $F_+(x)\log(1/F_+(x))\int_m^x\rho(s)^{-1}\,ds$ is bounded with the similar argument for $ F_- $.
 Thus $\mu$ satisfies a logarithmic Sobolev inequality
\cite[Theorem~1]{BR}, and \cite[Theorem~6]{Gross} gives
$\|P_{t_p}\|_{L^p\to L^2}\leq1$ for some $t_p<\infty$; take $t_2=0$.
Since $V''\geq-C_0$, comparison gives
$0\leq e^{tS}G\leq e^{C_0t}P_tG$ for $G\geq0$.
The spectral bound for $S$ in \eqref{eq:Sgap} therefore yields
\begin{equation}
 \|e^{tS}b'\|_2
 \leq e^{-\gamma_h(t-t_p)}e^{C_0t_p}\|b'\|_p,
 \qquad t\geq t_p.
 \label{eq:app-gain}
\end{equation}
Equations \eqref{eq:app-holder}--\eqref{eq:app-gain} prove \eqref{eq:app-result}
for $t\geq t_p$. For $t\leq t_p$, contractivity, the pure-state
trace-distance formula and Cauchy--Schwarz give
\[
 \|e^{t\cL}T-\Pi\|_{\Lone}
 \leq2\sqrt{1-\tr(\Pi T)}
 \leq2\sqrt{E(T)/\lambda_1(h)},
\]
using $H\geq\lambda_1(h)(I-\Pi)$. Since $\|d\|_q\geq1$, this covers
short times. For $q=\infty$, take $p=2$, $t_p=0$; the constant is $6\sqrt2$.

Finally, \eqref{eq:E2F} and $d\in L^q(\mu)$ show that 
 $F\in H^1(\mu;\ell^2)\cap L^{2q}(\mu;\ell^2)$.
For finite $q$, component and spatial cutoffs, then mollification and
normalization, give smooth finite-dimensional $F_n$ converging in both
spaces, with $\|F_n\|_2=1$. Thus
$\|T_n-T\|_{\Lone}\leq2\|F_n-F\|_2\to0$, and the energies and
diagonal norms converge. Contractivity passes \eqref{eq:app-result} to $T$.
For $q=\infty$, the same approximation gives
$\|F_n\|_\infty\leq(1+o(1))\|F\|_\infty$.
\end{proof}

\end{document}